\documentclass[11pt]{article}
\usepackage[letterpaper,
            left=1in,
            right=1in,
            top=1in,
            bottom=1in,
            footskip=.25in]{geometry}
\usepackage{float}
\usepackage{algorithm}
\usepackage{algorithmicx}
\usepackage[noend]{algpseudocode}
\usepackage{complexity}
\usepackage{graphics}
\usepackage{array}
\usepackage{changepage}
\usepackage{placeins}
\usepackage{multirow}
\usepackage{booktabs}
\usepackage{amsthm,amsmath}
\usepackage{hyperref}
\usepackage[capitalise]{cleveref}

\newclass{\USPACE}{USPACE}
\newclass{\TIUSP}{TIME{-}USPACE}
\newclass{\pUSP}{poly{-}USPACE}
\newclass{\qNC}{quasi{-}\NC}
\newclass{\coUL}{coUL}
\newclass{\NLC}{\NL{-}complete}
\newlang{\UTS}{Seg}
\newlang{\BTS}{BiSeg}
\newlang{\Reach}{Reach}
\newlang{\kReach}{k{-}Reach}

\newtheorem{theorem}{Theorem}[section]

\newtheorem{lemma}[theorem]{Lemma}
\newtheorem{remark}{Remark}

\usepackage{amssymb}
\title{Trading Determinism for Time: The $\kReach$ Problem}
\author{Ronak Bhadra, Raghunath Tewari}

\begin{document}
\maketitle




\begin{abstract}
Kallampally and Tewari showed in 2016 that there can be a trade-off between determinism and time in space-bounded computations. This they did by describing an unambiguous non-deterministic algorithm to solve Directed Graph Reachability that requires $O(\log^2 n)$ space and simultaneously runs in polynomial time. Savitch's 1970 algorithm that solves the same problem deterministically also requires $O(\log^2 n)$ space but doesn't guarantee polynomial running time and hence the trade off. We describe a new problem for which we can show a similar trade off between determinism and time.

We consider a collection $P$ of $f$ directed paths. We show that the problem of finding reachability from one vertex to another in the union $G$ of these path graphs via a path that switches amongst the paths in $P$ at most $k$ times can be solved in $O(k\log f+\log n)$ space but the algorithm doesn't guarantee polynomial runtime. On the other hand, we also show that the same problem can be solved by an unambiguous non-deterministic algorithm that simultaneously runs in $O(k\log f+\log n)$ space and polynomial time. Since these two algorithms are not dependent on Savitch, therefore this example sheds new light on how such a trade off between determinism and time happens in space-bounded computations and makes the phenomenon less elusive.
\end{abstract}


\section{Introduction}
\label{intro}

Unambiguous computations are a restriction of nondeterministic computations, where the nondeterministic machine accepts its input along at most one computation path. In other words it is a nondeterministic machine such that, on an instance belonging to the language, the machine has exactly one accepting computation, and on an instance not belonging to the language, the machine has no accepting computations. By definition, unambiguous computations lie between general nondeterministic computations and deterministic computations. It is an important question in computational complexity theory, whether any of these two containments is proper or not. In this paper, we study unambiguity in the context of space bounded computations. Unambiguous logspace (\UL) is a subclass of {\NL}, consisting of problems decidable by an {\NL} machine that has at most one accepting computation on all inputs. The class {\UL} was first defined and studied by Buntrock et al. \cite{BJLR91} and subsequently studied  by \`{A}lvarez and Jenner \cite{AJ93}. 

Reinhardt and Allender showed that {\NL} and {\UL} are equal in a non-uniform setting \cite{RA00}. Now whether the classes are equal uniformly or not, is an important open question. In a subsequent paper, Allender et al. showed that, under a reasonable hardness assumption that deterministic linear space has functions that can not be computed by circuits of size $2^{\epsilon n}$, we obtain $\NL=\UL$ \cite{ARZ99}. Therefore it is very much likely that that the two classes are equal.

In their seminal paper, Reinhardt and Allender also gave a way of showing that directed graph reachability is in {\UL} \cite{RA00}. They showed that, if for a class of graphs, an efficient edge weight function can be designed, with respect to which the minimum weight path between any pair of vertices is unique in the graph, then the reachability problem for that class of graphs is in {\UL}. Since the result of Reinhardt and Allender, there has been significant progress on the {\NL} vs {\UL} problem. For various class of graphs such as planar graphs \cite{BTV09,TV12}, bounded genus graphs \cite{KV10,DKTV11,GST20}, $K_{3,3}$ and $K_5$ minor-free graphs \cite{TW09,AGGT16}, graphs with  polynomially many paths from the start vertex to every other vertex \cite{PTV12}, the reachability problem has been shown to be in {\UL}.

Kallampally and Tewari in 2016 showed that any problem in {\NL} is decidable by an unambiguous algorithm running in polynomial time and using $O(\log^2 n)$ space \cite{KT16}. The space bound was subsequently improved to $O(\log^{1.5} n)$ in a later paper \cite{MP19}. However if we restrict ourselves to deterministic computation, then the best space upper bound on reachability is $O(\log^2 n)$ due to Savitch \cite{savitch} (however the time required by Savitch's algorithm is quasipolynomial). Thus, there happens to be a trade-off between determinism and time.

In this work, we describe a new problem called $\kReach$, having two parameters $k$ and $f$, for which we can show a similar trade-off between determinism and time. For a collection of directed path graphs $P$, we call the union graph of these path graphs as the \emph{union graph of $P$}. Given a collection $P$ of $f$ directed paths, the $\kReach$ problem is to decide whether there exists a path from one vertex to another in the union graph $G$ of $P$ that switches amongst the paths in $P$ at most $k$ times. The $\kReach$ problem can arise in many natural scenarios. Say we consider the network of trains connecting a bunch a cities. We can represent the collection of train routes as a collection $P$ of directed paths, where the cities are represented as nodes. Asking the question whether one can travel from one city to another by switching trains at most $k$ times, is akin to deciding $\kReach$ in $P$.

\subsection{Our Results}
In \cref{sec:upper}, we show that $\kReach$ can be solved deterministically using $O(f\log n)$ space in polynomial time but this bound does not incorporate the parameter $k$. In \cref{sec:Extension}, we show that the $\kReach$ problem is logspace reducible to the problem of detecting whether there is a $k$-length path from one vertex to another in a directed graph. Hence, this problem can be solved deterministically in $O(\log n\log k)$ space using Savitch \cite{savitch} but the algorithm doesn't guarantee polynomial runtime. On the other hand, this problem can also be solved by an unambiguous non-deterministic machine using $O(\log n\sqrt{\log k})$ space in polynomial time\cite{MP19}. However, these bounds don't incorporate the parameter $f$.  

In \cref{sec:kreach}, we show that the problem of $\kReach$ can be solved deterministically using $O(k\log f+\log n)$ space but our algorithm doesn't guarantee polynomial time. We then show that the same problem can be solved by an unambiguous non-deterministic machine using $O(k\log f+\log n)$ space in polynomial time. The space bound achieved here incorporates both the parameters $k$ and $f$. Finally, in \cref{sec:compare}, we draw a comparison among the different bounds on the $\kReach$ problem. 

\section{Upper Bound in terms of $f$}
\label{sec:upper}

In this section, we give an algorithm for deciding $\kReach$ in $P$, a collection of $f$ directed paths, that runs in space $O(f\log n)$ and polynomial time, where $n$ is the number of vertices in the union graph $G$ of $P$. Hence we get the following theorem.

\begin{theorem}
\label{thm:main1}
Given a collection of paths $P$ consisting of $f$ directed paths, $\kReach$ in $P$ can be decided in $O(f\log n)$ space and polynomial time.
\end{theorem}

\begin{proof}
We give an algorithm to decide reachability from one vertex $s$ to another vertex $t$ in the union graph $G$ of $P$. 

Whenever we will refer to the earliest vertex in a path $P_i$ having some property, we mean the vertex closest to the source vertex in $P_i$ having the said property.

The algorithm is as follows.

\begin{adjustwidth}{30pt}{10pt}
\begin{enumerate}

\item Keep two registers $c_i$ and $d_i$ corresponding to each directed path $P_i$ in $P$. (Since there are $f$ directed paths in $P$, therefore we need to keep $2f$ registers.)

\item Check if $t$ comes after $s$ in any of the directed paths in $P$. If yes, then halt and declare $t$ to be reachable from $s$. If not, then proceed to the next step.

\item For each directed path $P_i$ in $P$, check if $s$ is present in $P_i$. If yes, then store the index of $s$ in register $c_i$. Otherwise, initialize $c_i$ as NULL. Initialize $d_i$ as NULL for all $i$. Let us call the vertex indicated by the label in $c_i$ as $P_i[c_i]$. If $c_i$ is NULL, then $P_i[c_i]$ is also considered to be NULL.

\item For all $1\leq i\leq f$, check if $t$ occurs after $P_i[c_i]$ in path $P_i$. If yes, then halt and declare $t$ to be reachable from $s$. If not, then proceed to the next step. (If $P_i[c_i]$ is NULL, the answer to the above question is NO by default.)

\item For each directed path $P_j$ in $P$, find the earliest vertex $v_j$ in $P_j$, such that $v_j$ occurs after $P_i[c_i]$ (vertex indicated by the label in register $c_i$) in path $P_i$ for some $i$. Store the index of $v_j$ in register $d_j$. 

\item For all $i$, update $c_i=d_i$.

\item Repeat steps 4-6 until $t$ is found to be reachable from $s$ or values of $c_i$'s don't change anymore.

\end{enumerate}
\end{adjustwidth}

The algorithm uses a BFS kind of approach. The algorithm finds out all the vertices in $G$ that are $l$-reachable from $s$ within its $l$ iterations. In other words, any vertex that is $l$-reachable from $s$ will be detected as such by the algorithm within its first $l$ iterations. However, unlike standard BFS, we don't need to maintain a bit for every vertex indicating whether the vertex has been found to be reachable from $s$ or not. Rather, it suffices to maintain only the index of the earliest vertex in each path that is reachable from $s$. This is because a vertex $v$ in a path $p$ can be reachable from $s$ if and only if $v$ lies after the earliest vertex in $p$ that is reachable from $s$. Thus, we can check whether any vertex $v$ in $p$ is reachable from $s$ or not in logspace by simply traversing $p$.

In its $l$-th iteration, the algorithm finds out the earliest vertex $v_i$ in each path $p_i$ (for all $i$), that is $(l-1)$-reachable from $s$. All vertices that occur after vertex $v_i$ in path $p_i$ are $l$-reachable from $s$. The algorithm then checks if $t$ occurs after $v_i$ in path $p_i$ (for any $i$) or not. Thus, by iterating sufficient number of times, the algorithm is guaranteed to detect if $t$ is reachable from $s$. The algorithm halts when there is no change in the earliest reachable vertices for any of the paths, thus indicating that all vertices that are reachable from $s$ have already been detected.         

Hence, running the algorithm for a sufficient (at most $k$) number of iterations ensures that $t$ will be declared reachable from $s$ if $t$ is indeed $k$-reachable from $s$. If $t$ is not $k$-reachable from $s$, the algorithm can never declare $t$ to be $k$-reachable from $s$.

In this algorithm, we keep track of the earliest vertex reachable from $s$ in each of $f$ directed paths in $P$. We require $O(\log n)$ space to store the information of each vertex. Therefore we can store $f$ of them in $O(f\log n)$ space. Our algorithm works in deterministic $O(f\log n)$ space.

Each iteration of the algorithm takes polynomial (in $n$ and $f$) time. The algorithm can have $O(k)$ iterations in the worst case. The value of $f$ is $O(n^2)$ and that of $k$ is $O(n)$. Hence, the algorithm runs in polynomial time for all values of $k$ and $f$.
\end{proof}

\section{Upper Bounds on {\kReach} in terms of $k$}
\label{sec:Extension}

In this section, we show how the known upper bound for deciding reachability in a graph using unambiguous nondeterministic \cite{MP19} and deterministic computation \cite{savitch} can be extended to give an upper bound for the $\kReach$ problem.

First, we give a reduction from the $\kReach$ problem to the problem of deciding reachability in a layered digraph having $k+1$ layers. 

\begin{lemma}
\label{thm:dtold}
Given a collection of directed paths $P$, the $\kReach$ problem in $P$ is logspace reducible to the problem of deciding reachability in a layered digraph having $k+1$ layers.
\end{lemma}

\begin{proof}
Let $G$ be the union graph of $P$. Construct a layered digraph $L$ having $k+1$ layers, such that each layer has $n$ nodes. In total $L$ has $n(k+1)$ nodes. Let $x_i$ denote the $i$th node in $G$ and $y_{ij}$ denote the $j$th node in $i$th layer in $L$. For every pair of vertices $(x_p,x_q)$ in $G$, we have the following edges in $L$- $(y_{1p},y_{2q}), (y_{2p},y_{3q}), (y_{3p},y_{4q}),\dots, (y_{(k)p},y_{(k+1)q})$, if and only if $x_q$ comes after $x_p$ in some path in $P$. For every vertex ${x_m}$ in $G$, we include the following edges in $L$- $(y_{1m},y_{2m}),(y_{2k},y_{3k}),\dots,(y_{(n-1)k},y_{nk})$. The construction of $L$ from $G$ can be done in logspace. The problem of deciding the reachability from vertex $x_u$ to vertex $x_v$ in $G$ via a path that switches amongst the paths in $P$ at most $k$ times now reduces to the problem of deciding the reachability from vertex $y_{1x}$ to vertex $y_{(k+1)y}$ in the layered digraph $L$. Thus, $x_v$ is $k$-reachable from $x_u$ in $P$ if and only if $y_{(k+1)v}$ is reachable from $y_{1u}$ in $L$.
\end{proof}

Reachability in a layered digraph having $k$ layers can be solved in $O(\log k\log n)$ deterministic space\cite{savitch} and also in unambiguous space $O(\sqrt{\log k}\log n)$ in polynomial time\cite{MP19}. Thus, we get the following bounds on the space complexity of $\kReach$.

\begin{theorem}
Given a collection $P$ of $f$ directed paths, $\kReach$ can be decided in $P$ deterministically using $O(\log n\log k)$ space.
\end{theorem}

\begin{theorem}
\label{thm:kreach2}
Given a collection $P$ of $f$ directed paths, $\kReach$ can be decided in $O(\log n\sqrt{\log k})$ space and polynomial time by an unambiguous (and co-unambiguous) nondeterministic Turing machine. 
\end{theorem}
 
\section{Upper Bounds on {\kReach} in terms of both $f$ and $k$ together}
\label{sec:kreach}

In \cref{sec:Extension}, we show the upper bounds on $\kReach$ that can be derived from known results in both deterministic \cite{savitch} and unambiguous setting \cite{MP19}. However, this is not apparent how to take into account the parameter $f$ into consideration. In \cref{sec: Algorithm 2}, we provide an upper bound on the complexity of the $\kReach$ problem in deterministic setting, in terms of both the parameters $f$ and $k$. However, it doesn't guarantee polynomial runtime. In \cref{sec: Algorithm 3}, we show that in unambiguous setting, the same space complexity upper bound as in \cref{sec: Algorithm 2} can be achieved in polynomial runtime. The bounds in \cref{sec: Algorithm 2} and \cref{sec: Algorithm 3} are kind of a trade-off between runtime and determinism, also shown earlier in \cite{KT16}.

First we define a few concepts which will be utilized in the rest of the section. Let $P$ be a collection of $f$ directed paths and $G$ be the union graph of $P$. For a path $p$ from vertex $s$ to vertex $t$ in $G$, we define the function $\text{switch}(p)$ to be the number of switches $p$ makes amongst the paths in $P$. For a vertex $x$, we define $d(x)$ to be the minimum value of switch($p$) such that $p$ is a path from $s$ to $x$ in $G$. We also call $d(x)$ to be the distance of $x$ from $s$.

We also define the function seq($p$) as a number, which when represented in base $f$ consists of $\text{switch}(p)$ digits. The $i$-th digit in seq($p$), when represented in base $f$, is the index of the path in $P$ that the path $p$ utilizes during its $i$-th switch among the paths in $P$. We require $O(k\log f)$ bits to represent seq($p$) for any path $p$.

\subsection{Deterministic Algorithm for \texorpdfstring{$\kReach$}{k-Reach}}
\label{sec: Algorithm 2}

We show that $k$-reachability in a collection $P$ of $f$ directed paths can be solved deterministically using $O(k\log f+\log n)$ space. We use a DFS sort of approach to achieve this bound. We iterate over all permutations of $k$ paths from $P$ and check for each permutation $p_1p_2\dots p_k$, whether there exists a path from $s$ to $t$ in the union graph $G$ of $P$, such that seq($p$)=$p_1p_2\dots p_k$. We give a deterministic polytime routine, which for a given value $w$, checks for the existence of a path $p$ such that seq($p$)=$w$.

\begin{theorem}
\label{thm:kreach1}
Given a collection $P$ of $f$ directed paths, $k$-reachability in $P$ can be decided deterministically in $O(k\log f+\log n)$ space.
\end{theorem}

\begin{proof}
We keep a register $c$ consisting of $k\log f$ bits, which is structured as $k$ groups of $\log f$ bits each. We call a particular configuration of $c$ to be \emph{valid} if each group of $\log f$ bits in $c$ is a label for some path in $P$. We call the $i$-th group of $\log f$ bits in $c$ as $c_i$. We call the path in $P$ corresponding to some label $l$ as $P[l]$.

The algorithm is as follows.

\begin{adjustwidth}{30pt}{10pt}

\begin{enumerate}

\item  Check if $t$ comes after $s$ in any of the paths in $P$. If yes, then halt and declare $t$ to be reachable from $s$. Otherwise, repeat steps 2 to 5 for all possible valid configurations of the register $c$.

\item Initialize register $d=s$ (that is, $d$ contains the label of vertex $s$) and register $i=0$. Let $v_d$ be the vertex indicated by the label in $d$. We denote the path in $P$ indicated by the label in register $c_i$ as $P[c_i]$.

\item If $i=0$, then check if $s$ is in the path $P[c_0]$. If $i>0$, then find (if any) the earliest vertex in the path $P[c_i]$ that occurs after $v_d$ in the path $P[c_{i-1}]$. Update the label of this vertex in register $d$. In both the above cases, if the step is not successful, then break and try the next valid configuration of $c$ (step 1). Otherwise, check if $t$ comes after $v_d$ in path $P[c_i]$. If yes, then halt and accept. Else, update $i=i+1$ and proceed. 

\item Repeat steps 2 and 3 $k$ times.

\item If $t$ is not found to be reachable from $s$ for any of the possible valid configurations of $c$, then halt and reject.
\end{enumerate}

\end{adjustwidth}

The above algorithm accepts if and only if $t$ is $k$-reachable from $s$ in $P$. It is easy to see that the algorithm accepts only when there is a path from $s$ to $t$ that switches amongst the paths in $P$ at most $k$ times. For the other direction, let us assume that there exists a path $p$ from $s$ to $t$ that switches amongst the paths in $P$ at most $k$ times. In the iteration of the algorithm when seq($p$) occurs as a prefix in the content of the register $c$, the algorithm must accept. This is because step 3 of the algorithm always finds the earliest vertex in a path $P[c_i]$ that is $i-1$-reachable from $s$ via a path whose seq is $c_0c_1\dots c_{i-1}$. By doing this, step 3 of the algorithm discovers all vertices (the vertices coming after the earliest vertex) in path $P[c_i]$  that are $(i-1)$-reachable from $s$ via a path whose seq is $c_0c_1\dots c_{i-1}$. 

We see that this algorithm takes $O(k\log f + \log n)$ space, because it requires $O(k\log f)$ space for the register $c$ and $O(\log n)$ space for the rest of the registers.
\end{proof}

\begin{remark}
The algorithm in the proof of \cref{thm:kreach1} is not guaranteed to work in polynomial time for all $k$ and $f$ since the algorithm iterates over all $f^k$ possible valid configurations of the register $c$.
\end{remark}

\subsection{Unambiguous Nondeterministic Algorithm for \texorpdfstring{$\kReach$}{k-Reach}}
\label{sec: Algorithm 3}
We now provide an unambiguous nondeterministic algorithm for $\kReach$ problem, which works in polynomial time. This is a modified version of the double-inductive counting algorithm provided by Reinhardt and Allender \cite{RA00} to decide reachability in a min-unique graph in unambiguous logspace. 

\begin{theorem}
\label{thm:kreach3}
Given a collection $P$ of $f$ directed paths, $k$-reachability in $P$ can be decided in $O(k\log f+\log n)$ space by an unambiguous (and co-unambiguous) nondeterministic machine in polynomial time.
\end{theorem}

\begin{proof}
Let $s$ and $t$ be two vertices in $P$. Let $c_h$ be the number of vertices that are reachable from $s$ via paths that have switch$\leq h$. The weight of any path $p$ from $s$ to some vertex $v$ in $G$ is defined to be equal to seq($p$). Let $\Sigma_h$ be the sum of the weights of the minimum weight paths to all vertices for which there is a path from $s$ that has switch$\leq h$.

The algorithm has three subroutines. The first subroutine takes as input $c_h$, $\Sigma_h$ and a vertex $v$ and decides unambiguous nondeterministically if d$(v)\leq h$. It also returns the weight of the minimum weight path from $s$ to $v$ if d$(v)\leq h$. The second subroutine takes $c_{h-1}$ and $\Sigma_{h-1}$ and computes $c_h$ and $\Sigma_h$ using subroutine 1. The third subroutine inductively computes $c_h$ and $\Sigma_h$ for all $1\leq h\leq k$, using the second subroutine, and finally invokes the first subroutine to decide if d$(t)\leq k$. The inductive counting happens over the value of the switch of the paths from $s$ to other vertices. That is, at each step, we construct a ball of radius $h$ around $s$ and compute using $c_{h-1},\Sigma_{h-1}$, the number of vertices, $c_h$, that are at a distance of $h$ from $s$ and the sum, $\Sigma_k$, of the weights of the minimum length paths from $s$ to all those vertices. Finally, we invoke the first subroutine to decide if $t$ is reachable from $s$ via a path of switch $k$.

The Reinhardt-Allender Algorithm requires the weight function to be a min-unique weight function (that is there should be a unique minimum weight path under this weight function) in order to work. This may not always be true for our weight function. There may be multiple minimum weight paths in the graph according to our weight function. However, we have a way to get around that.  Given a vertex $v$, we nondeterministically guess the weight of a path from $s$ to $v$. We then follow a deterministic procedure to determine if there is any path of that guessed weight from $s$ to $v$ in a similar manner as in steps 2-4 in the proof of \cref{thm:kreach1}. This checking can be done deterministically in polynomial time using $O(h\log f+\log n)$ space, as already shown in the proof of \cref{thm:kreach1}. Thus, even if there are multiple paths of a particular weight, our algorithm still works unambiguously. Also, each of the subroutines work in polynomial time, irrespective of the values of $f$ and $k$. Thus, the overall Algorithm which calls the subroutine 2 some $k$ number of times also works in polynomial time.  

\begin{algorithm}[H]
\caption{Determine if d$(v)\leq h$}
\begin{algorithmic}[1]
\Function{Weight}{$v,h,c_{h},\Sigma_{h}$}

\For{each $x\in V$ }
\State nondeterministically guess if d$(x)\leq h$
\If {guess is yes}
\State guess a value of d$(x)$, say $l\leq h$, and also a sequence of indices $i_1,i_2,\dots,i_l$ and check if there is a path $p$ from $s$ to $x$ such that seq($p$)= $i_1,i_2,\dots,i_l$
\If {guess is correct}
\State $count=count+1$
\State $sum=sum+i_1*f^{l-1}+i_2*f^{l-2}+\dots+i_l$
\If {$x==v$}
\State $path.to.v$=$true$, $\sigma=i_1*f^{l-1}+i_2*f^{l-2}+\dots+i_l$
\EndIf
\Else
\State \Return "?" 
\EndIf

\EndIf

\EndFor
\If {$count==c_h$ and $sum==\Sigma_h$}
\State \Return $path.to.v$, $\sigma$
\Else
\State \Return "?"
\EndIf
\EndFunction
\end{algorithmic}
\end{algorithm}
     
\begin{algorithm}[H]
\caption{Computing $c_h$ and $\Sigma_h$}

\begin{algorithmic}[1] 
\Require {$h,c_{h-1},\Sigma_{h-1}$}

\State Initialize $c_h\leftarrow c_{h-1}$, $\Sigma_h\leftarrow \Sigma_{h-1}$

\For{each $v\in V$ }
\State Initialize $flag\leftarrow0$, $\sigma\leftarrow\infty$
\State $path.to.v,z=$\Call{Weight}{$v,h-1,c_{h-1},\Sigma_{h-1}$}
\If {$path.to.v==false$}
\For{each $x$ in $G$}
\State $path.to.x,w=$\Call{Weight}{$x,h-1,c_{h-1},\Sigma_{h-1}$}
\If {$path.to.x==true$}
\If{$v$ is reachable from $x$ via a path $p$ such that switch($p$)=0}
\State $flag=1$
\If {$w*f+i<\sigma$}
\State $\sigma=w*f+i$, where $i$ is the smallest index of the path in $P$ via which $v$ is reachable from $x$.
\EndIf
\EndIf
\EndIf
\EndFor
\EndIf
\If {$flag>0$}
\State $c_h=c_h+1$
\State $\Sigma_h=\Sigma_h+\sigma$
\EndIf

\EndFor

\Return {} $c_h$, $\Sigma_h$

\end{algorithmic}
\end{algorithm}

\begin{algorithm}[H]
\caption{Determining if there exists a path from $s$ to $t$ in $G$}
\renewcommand{\Require}{\textbf{Input:}}
\renewcommand{\Ensure}{\textbf{Output:}}
\begin{algorithmic}[1] 

\State Initialize $c_0\leftarrow 1$, $\Sigma_0\leftarrow 0, h \leftarrow 0$

\For{$h=1\dots k$ }
\State Compute $c_h$ and $\Sigma_h$ by invoking Algorithm 2 on $(h,c_{h-1},\Sigma_{h-1})$

\EndFor

\State Invoke Algorithm 1 on $(t,k,c_k,\Sigma_k)$ and return its value

\end{algorithmic}
\end{algorithm}
 
\end{proof}

\section{Comparison among the Upper Bounds for $\kReach$}
\label{sec:compare}
The bounds in \cref{sec:kreach} for the $\kReach$ problem are worse than the bounds in \cref{sec:upper} and \cref{sec:Extension} for most values of the parameters $f$ and $k$. However, for certain restricted settings of the parameters $k$ and $f$, the bounds in \cref{sec:kreach} perform better than the other bounds and have interesting implications. In this section, we will highlight two such cases.

\begin{itemize}

\item \underline{\textbf{Case 1}}:- $k=\sqrt{\log n}$, $f=2^{\sqrt{\log n}}$

The comparison of the different bounds in this case is summarized in \cref{table:tab1}.

\begin{table}[h!]
\begin{tabular}{|c|ll|l|l}
\cline{1-4}
\textbf{Setting} & \multicolumn{2}{c|}{\textbf{\begin{tabular}[c]{@{}c@{}}Space Complexity\\ in terms of $f,k,n$\end{tabular}}} & \multicolumn{1}{c|}{\textbf{\begin{tabular}[c]{@{}c@{}}Space Complexity\\ in terms of $n$\end{tabular}}} &  \\ \cline{1-4}
Deterministic (\cref{sec:Extension}) & \multicolumn{2}{l|}{$O(\log n\log k)$} & $O(\log n\log\log n)$ &  \\ \cline{1-4}
Unambiguous Polytime (\cref{sec:Extension}) & \multicolumn{2}{l|}{$O(\log n\sqrt{\log k})$} & $O(\log n\sqrt{\log\log n})$ &  \\ \cline{1-4}
\begin{tabular}[c]{@{}c@{}}Deterministic Polytime\\ (\cref{sec:upper})\end{tabular} & \multicolumn{2}{l|}{$O(f\log n)$} & $O(2^{\sqrt{\log n}}\log n)$ &  \\ \cline{1-4}
\begin{tabular}[c]{@{}c@{}}Deterministic (Not polytime)/\\ Unambiguous Polytime\\ (\cref{sec:kreach})\end{tabular} & \multicolumn{2}{l|}{$O(k\log f+\log n)$} & $O(\log n)$ &  \\ \cline{1-4}
\end{tabular}
\caption{}
\label{table:tab1}
\end{table}

In this case, the bound from \cref{sec:kreach} evaluates to $O(\log n)$. All the other bounds remain super-logarithmic. However, since the space consumption by our deterministic routine is $O(\log n)$, it also runs in polytime and hence, no separate unambiguous routine is necessary.

\item \underline{\textbf{Case 2}}:- $k=\log n(\log\log n)^{0.2}$, $f=2^{(\log\log n)^{0.2}}$

The comparison of the different bounds in this case is summarized in \cref{table:tab2}.

 \begin{table}[]
 \label{tab:2}
\begin{tabular}{|c|ll|l|l}
\cline{1-4}
\textbf{Setting} & \multicolumn{2}{c|}{\textbf{\begin{tabular}[c]{@{}c@{}}Space Complexity\\ in terms of $f,k,n$\end{tabular}}} & \multicolumn{1}{c|}{\textbf{\begin{tabular}[c]{@{}c@{}}Space Complexity\\ in terms of $n$\end{tabular}}} &  \\ \cline{1-4}
Deterministic (\cref{sec:Extension}) & \multicolumn{2}{l|}{$O(\log n\log k)$} & $O(\log n\log\log n)$ &  \\ \cline{1-4}
Unambiguous Polytime (\cref{sec:Extension}) & \multicolumn{2}{l|}{$O(\log n\sqrt{\log k})$} & $O(\log n\sqrt{\log\log n})$ &  \\ \cline{1-4}
\begin{tabular}[c]{@{}c@{}}Deterministic Polytime\\ (\cref{sec:upper})\end{tabular} & \multicolumn{2}{l|}{$O(f\log n)$} & $O(2^{(\log\log n)^{0.2}}\log n)$ &  \\ \cline{1-4}
\begin{tabular}[c]{@{}c@{}}Deterministic (Not polytime)/\\ Unambiguous Polytime\\ (\cref{sec:kreach})\end{tabular} & \multicolumn{2}{l|}{$O(k\log f+\log n)$} & $O(\log n(\log\log n)^{0.4})$ &  \\ \cline{1-4}
\end{tabular}
\caption{}
\label{table:tab2}
\end{table}

In this case also, the bound from \cref{sec:kreach} performs better than all the other bounds. Moreover, since the space consumption by our deterministic routine is superlogarithmic, it does not run in polytime and hence, our unambiguous routine finds use here which runs in polytime with the same space bound.

\end{itemize}




\bibliography{references}

\begin{thebibliography}{10}

\bibitem{ARZ99}
Eric Allender, Klaus Reinhardt, and Shiyu Zhou.
\newblock Isolation, matching, and counting: Uniform and nonuniform upper bounds.
\newblock {\em Journal of Computer and System Sciences}, 59:164--181, 1999.

\bibitem{AJ93}
Carme \`{A}lvarez and Birgit Jenner.
\newblock A very hard log-space counting class.
\newblock {\em Theoretical Computer Science}, 107:3--30, 1993.

\bibitem{AGGT16}
Rahul Arora, Ashu Gupta, Rohit Gurjar, and Raghunath Tewari.
\newblock Derandomizing isolation lemma for ${K}_{3,3}$-free and ${K}_5$-free bipartite graphs.
\newblock In {\em 33rd Symposium on Theoretical Aspects of Computer Science, {STACS} 2016, February 17-20, 2016, Orl{\'{e}}ans, France}, pages 10:1--10:15, 2016.

\bibitem{BTV09}
Chris Bourke, Raghunath Tewari, and N.~V. Vinodchandran.
\newblock Directed planar reachability is in unambiguous log-space.
\newblock {\em ACM Transactions on Computation Theory}, 1(1):1--17, 2009.
\newblock \href {https://doi.org/http://doi.acm.org/10.1145/1490270.1490274} {\path{doi:http://doi.acm.org/10.1145/1490270.1490274}}.

\bibitem{BJLR91}
Gerhard Buntrock, Birgit Jenner, Klaus-J\"{o}rn Lange, and Peter Rossmanith.
\newblock Unambiguity and fewness for logarithmic space.
\newblock In {\em Proceedings of the 8th International Conference on Fundamentals of Computation Theory (FCT'91)}, Volume 529 Lecture Notes in Computer Science, pages 168--179. Springer-Verlag, 1991.

\bibitem{DKTV11}
Samir Datta, Raghav Kulkarni, Raghunath Tewari, and N.V. Vinodchandran.
\newblock Space complexity of perfect matching in bounded genus bipartite graphs.
\newblock {\em Journal of Computer and System Sciences}, 78(3):765 -- 779, 2012.
\newblock In Commemoration of Amir Pnueli.
\newblock URL: \url{http://www.sciencedirect.com/science/article/pii/S002200001100136X}, \href {https://doi.org/10.1016/j.jcss.2011.11.002} {\path{doi:10.1016/j.jcss.2011.11.002}}.

\bibitem{GST20}
Chetan Gupta, Vimal~Raj Sharma, and Raghunath Tewari.
\newblock Efficient isolation of perfect matching in o(log n) genus bipartite graphs.
\newblock In Javier Esparza and Daniel Kr{\'{a}}l', editors, {\em 45th International Symposium on Mathematical Foundations of Computer Science, {MFCS} 2020, August 24-28, 2020, Prague, Czech Republic}, volume 170 of {\em LIPIcs}, pages 43:1--43:13. Schloss Dagstuhl - Leibniz-Zentrum f{\"{u}}r Informatik, 2020.
\newblock \href {https://doi.org/10.4230/LIPIcs.MFCS.2020.43} {\path{doi:10.4230/LIPIcs.MFCS.2020.43}}.

\bibitem{KT16}
Vivek Anand~T. Kallampally and Raghunath Tewari.
\newblock Trading determinism for time in space bounded computations.
\newblock In {\em 41st International Symposium on Mathematical Foundations of Computer Science, {MFCS} 2016, August 22-26, 2016 - Krak{\'{o}}w, Poland}, pages 10:1--10:13, 2016.
\newblock \href {https://doi.org/10.4230/LIPIcs.MFCS.2016.10} {\path{doi:10.4230/LIPIcs.MFCS.2016.10}}.

\bibitem{KV10}
Jan Kyn\v{c}l and Tom\'{a}\v{s} Vysko\v{c}il.
\newblock Logspace reduction of directed reachability for bounded genus graphs to the planar case.
\newblock {\em ACM Transactions on Computation Theory}, 1(3):1--11, 2010.
\newblock \href {https://doi.org/http://doi.acm.org/10.1145/1714450.1714451} {\path{doi:http://doi.acm.org/10.1145/1714450.1714451}}.

\bibitem{PTV12}
Aduri Pavan, Raghunath Tewari, and N.~V. Vinodchandran.
\newblock On the power of unambiguity in log-space.
\newblock {\em Computational Complexity}, 21(4):643--670, 2012.
\newblock URL: \url{http://dx.doi.org/10.1007/s00037-012-0047-3}, \href {https://doi.org/10.1007/s00037-012-0047-3} {\path{doi:10.1007/s00037-012-0047-3}}.

\bibitem{RA00}
Klaus Reinhardt and Eric Allender.
\newblock Making nondeterminism unambiguous.
\newblock {\em {SIAM} J. Comput.}, 29(4):1118--1131, 2000.
\newblock URL: \url{http://dx.doi.org/10.1137/S0097539798339041}, \href {https://doi.org/10.1137/S0097539798339041} {\path{doi:10.1137/S0097539798339041}}.

\bibitem{savitch}
Walter~J. Savitch.
\newblock Relationships between nondeterministic and deterministic tape complexities.
\newblock {\em Journal of Computer and System Sciences}, 4(2):177--192, 1970.
\newblock \href {https://doi.org/http://dx.doi.org/10.1016/S0022-0000(70)80006-X} {\path{doi:http://dx.doi.org/10.1016/S0022-0000(70)80006-X}}.

\bibitem{TV12}
Raghunath Tewari and N.~V. Vinodchandran.
\newblock Green's theorem and isolation in planar graphs.
\newblock {\em Inf. Comput.}, 215:1--7, 2012.
\newblock \href {https://doi.org/10.1016/j.ic.2012.03.002} {\path{doi:10.1016/j.ic.2012.03.002}}.

\bibitem{TW09}
Thomas Thierauf and Fabian Wagner.
\newblock {R}eachability in ${K}_{3,3}$-free {G}raphs and ${K}_5$-free {G}raphs is in {U}nambiguous {L}og-{S}pace.
\newblock In {\em 17th International Conference on Foundations of Computation Theory (FCT)}, Lecture Notes in Computer Science 5699, pages 323--334. Springer-Verlag, 2009.

\bibitem{MP19}
Dieter van Melkebeek and Gautam Prakriya.
\newblock Derandomizing isolation in space-bounded settings.
\newblock {\em SIAM Journal on Computing}, 48(3):979--1021, 2019.
\newblock \href {https://arxiv.org/abs/https://doi.org/10.1137/17M1130538} {\path{arXiv:https://doi.org/10.1137/17M1130538}}, \href {https://doi.org/10.1137/17M1130538} {\path{doi:10.1137/17M1130538}}.

\end{thebibliography}







\end{document}